\documentclass[aps,pra,reprint,twocolumn,superscriptaddress,nofootinbib]{revtex4-2}
\usepackage{amsmath,amssymb,amsthm,mathtools}
\usepackage{physics}
\usepackage{comment}
\usepackage{bm}
\usepackage{microtype}
\usepackage[colorlinks=true,linkcolor=blue,citecolor=blue,urlcolor=blue]{hyperref}
\usepackage[capitalize,noabbrev]{cleveref}

\newtheorem{theorem}{Theorem}
\newtheorem{lemma}[theorem]{Lemma}
\newtheorem{prot}{Protocol}

\newtheorem{corollary}[theorem]{Corollary}
\newtheorem{definition}[theorem]{Definition}

\newtheorem{remark}[theorem]{Remark}

\crefname{theorem}{Theorem}{Theorems}
\crefname{lemma}{Lemma}{Lemmas}

\crefname{corollary}{corollary}{corollaries}
\crefname{definition}{Definition}{Definitions}
\crefname{assumption}{assumption}{assumptions}
\crefname{remark}{Remark}{Remarks}

\newcommand{\EprePA}{\mathbf{E}_\mathrm{prePA,EV}}
\newcommand{\dist}[2]{d\!\left(#1,#2\right)}

\newcommand{\bS}{\mathbf{S}}
\newcommand{\Cpre}{C_{\mathrm{prePA}}}
\newcommand{\CEV}{C_{\mathrm{EV}}}
\newcommand{\HEV}{H_{\mathrm{EV}}}
\newcommand{\FEV}{F_{\mathrm{EV}}}
\newcommand{\Efin}{E_{\mathrm{fin}}}
\newcommand{\Omlen}[1]{\Omega^{\mathrm{len}}_{#1}}
\newcommand{\Omfin}[1]{\Omega^{\mathrm{fin}}_{#1}}
\newcommand{\OmFK}[1]{\Omega^{(F_K)}_{#1}}
\newcommand{\Ommatch}{\Omega_{\mathrm{match}}}
\newcommand{\OmEVpost}{\Omega^\mathrm{comb}_{\mathrm{EV}}}
\newcommand{\Omacc}{\Omega_{\mathrm{acc}}}
\newcommand{\Omcoll}{\Omega_{\mathrm{coll}}}

\newcommand{\one}{\mathbf{1}}

\newcommand{\elloutput}{\ell_\mathrm{output}}
\newcommand{\ellEV}{\ell_\mathrm{EV}}
\newcommand{\ellprePA}{\ell_\mathrm{prePA}}

\newcommand{\HPA}{H_\mathrm{PA}}

\newcommand{\cobs}{c_\mathrm{obs}}

\begin{document}

\title{Combining Error Verification and Privacy Amplification in Quantum Key Distribution}

  \author{Devashish Tupkary}
\email{devashish.jayant.tupkary@umontreal.ca}
	 \affiliation{Institute for Quantum Computing and Department of Physics and Astronomy, University of Waterloo, Waterloo, Ontario, Canada, N2L 3G1}	
     \affiliation{Department of Computer Science and Operations Research,
Université de Montréal,
Montréal, Quebec, Canada, H3T 1J4}

 \author{Zhiyao Wang}
	 \affiliation{Institute for Quantum Computing and Department of Physics and Astronomy, University of Waterloo, Waterloo, Ontario, Canada, N2L 3G1}	
  \author{Ernest Y.-Z. Tan}
      \affiliation{Department of Physics, National University of Singapore, Singapore, Singapore, Singapore, 117542}

\begin{abstract}
Quantum Key Distribution (QKD) protocols typically involve 
two separate steps of error-verification and privacy amplification. Both steps involve the application of a hash function sampled from a suitable hash family to classical data. A potentially simpler alternative is to implement a single-step hashing procedure, announce and compare a fixed number of bits of its output for error-verification, and retain the remaining bits as the final output key. This reduces the complexity of classical postprocessing, the amount of classical communication, and the amount of randomness needed in the protocol. While this construction has appeared in the literature, a suitable mathematical justification is lacking. In this work, we show that the security of such protocols, under certain simple conditions, follows from a slightly modified notion of secrecy of the protocol without the final error-verification step. In particular, we show that for most entropy-based proofs, the combined protocol can be implemented with no change to the key rate, and a minor change to the security parameter.
\end{abstract}

\maketitle
\section{Introduction}\label{sec:introduction}

QKD protocols involve the transmission and
measurement of quantum states, followed by classical postprocessing. Towards the end of such protocols, the honest parties Alice and Bob hold ``pre-amplification strings" of equal length, which are expected (but not guaranteed)
to agree. To check their agreement, they apply a common hash function,
sampled independently from a suitable family, and compare
the resulting outputs over an authenticated classical channel\footnote{In this work, we assume that this means all messages sent are received correctly by the other party, some time after they are sent. Thus, we restrict attention to protocols in which Alice and Bob always agree on the output key length and abort status. For work that relaxes this assumption, see Ref.~\cite{tupkary2026authenticationsecurityproofsquantum}.}. If the
outputs differ, they abort the protocol. This step, called error
verification (EV), ensures that if the pre-amplification strings contain an error, then the protocol aborts with high probability; i.e., it ensures correctness. Eve may still retain some information about the pre-amplification strings. Alice and Bob therefore apply another common hash
function, chosen using fresh randomness and announced over the
authenticated classical channel, to obtain shorter strings. This step,
called privacy amplification (PA), is intended to ensure secrecy: Alice's
final key should be close to uniform and independent of Eve's registers. 

A natural implementation choice is to combine EV and PA into a single
hashing step. Alice and Bob apply the same hash function to their
respective strings. Alice announces a few bits of her output, and Bob checks whether they agree
with the corresponding bits of his output. If they agree, both parties
discard these bits and retain the remaining bits as their final keys.
Otherwise, they abort. 

This construction avoids a separate EV hash
evaluation and can reduce the randomness and public communication needed
to specify the hash functions.\footnote{Note that, a uniformly random binary
Toeplitz matrix mapping $n$ input bits to $m$ output
bits is specified by $n+m-1$ independent random bits,
which can be announced as the public seed.} 

The earliest explicit proposal for such a combined construction that we have
located is Ref.~\cite{HughesEtAl2002}, although the construction was neither
implemented nor analyzed there. Ref.~\cite{Fossier_2009} proposed and implemented
a similar procedure in which a randomly chosen subset of the final key
bits is compared and discarded. Ref.~\cite{ZhangEtAl2014} again proposed a
similar construction and analyzed its correctness, but did not analyze the
secrecy of the retained key conditioned on the verification test passing.
Ref.~\cite{MaroyEtAl2014} also proposed comparing bits of the PA output but did not analyze the secrecy of the retained key. We refer the reader to Ref.~\cite{Fung_2010} for an overview
of classical post-processing in QKD, and
Ref.~\cite{Tsurumaru_securityloophole2024} for a recent treatment of a subtlety in error-verification in security proofs based on the phase-error-correction
approach~\cite{koashi_simple_2009}.

Despite its simplicity, the combined protocol requires a security
argument that does not follow immediately from the usual properties of
PA. First, correctness now needs to hold when only a subset of output bits are compared. Second, even if Alice's
complete hash output is secret before verification, Bob's decision to
accept may reveal information about the retained bits. These issues are
discussed in Ref.~\cite[Section 3.4.2]{TupkaryEtAl2026}.

In this work, we give a simple analysis of the combined
protocol. In \cref{sec:setup} we setup notation and describe the combined protocol we analyze. In \cref{sec:results} we state our main result in the form of \cref{thm:main}, which states that under certain conditions, the security of the combined protocol follows from a slightly modified notion of secrecy of the protocol without the final EV step.  In \cref{subsec:impactonkeyrates} we apply this result and show that for most entropy based proofs, combining the EV and PA step causes no change to the secure key rate. In \cref{sec:conclusion} we present some concluding remarks.

\section{Notation and Protocol description}\label{sec:setup}

\subsection{Basic notation}\label{subsec:notation}

We use uppercase letters such as $X, A, B$ to denote registers. We use $S_=(A)$ to denote the set of density operators on the register $A$, and $S_{\leq}(A)$ to denote the set of subnormalized density operators on the register $A$.  We use $\norm{\rho_A}_1$ to denote the Schatten one-norm of $\rho_A$. We use $d(\rho,\sigma) \coloneq \frac{1}{2} \norm{\rho-\sigma}_1$ to denote the trace distance between two density operators. 

Classical registers are diagonal in some
fixed basis.  For a state $\rho \in S_{\leq}(CQ)$ classical on $C$, written in the form
$\rho_{CQ} = \sum_c \ketbra{c}{c} \otimes \omega_c$ 
for some $\omega_c \in S_{\leq}(Q)$,
and an event $\Omega$ defined on the register $C$, we will define a corresponding \emph{partial} state and \emph{conditional} state as, respectively,
\begin{align}
\rho_{\land\Omega} &\coloneq \sum_{c\in\Omega} \ketbra{c}{c} \otimes \omega_c, \qquad\qquad \\ \rho_{|\Omega} &\coloneq \frac{\tr{\rho}}{\tr{\rho_{\land\Omega}}} \rho_{\land\Omega} = \frac{
\sum_{c} \tr{\omega_c}
}{\sum_{c\in\Omega} \tr{\omega_c}} \rho_{\land\Omega} .
\end{align}
We refer to $\rho_{\wedge \Omega}$ as the state $\rho$ being partial\footnote{One could instead also think of it as being ``subnormalized conditioned'' on the event $\Omega$.} on $\Omega$.

\subsection{Combined EV--PA protocol}\label{subsec:protocol}

We now describe the combined protocol we analyze. All public communication during the protocol is authenticated, but known to the adversary.

\begin{prot}[QKD with Combined PA and EV] \label{prot:combined}
Following some QKD protocol, Alice and Bob possess classical
pre-amplification strings $\bS_A$, $\bS_B$ of the same length. They have determined $\ell_\mathrm{output}$ the length of output bits to be produced, which is stored in the register $F_K$. 
Eve holds a side-information register $\EprePA$ which may be correlated with the bit strings, and may include a copy of all public communication during the preceding steps (including $F_K$). 
\begin{enumerate}

\item \textbf{Fixing hash family.} Alice reads the value of $\ell_\mathrm{output}$ from the $F_K$ register, and the length $\ellprePA$ of strings stored in the $\bS_A,\bS_B$ registers. This fixes a universal$_2$ hash family from $\ellprePA$ bits to $(\elloutput+\ellEV)$ bits.
 
\item \textbf{Choosing hash function.} Alice samples the hash  seed $h$ from the specified distribution, 
\emph{independently} of all other registers, and announces it in the register $\HPA$. This fixes the function $f_h$ from the hash family.
 
\item \textbf{Applying hash function.}
Alice computes $f_h(\bS_A)$ and stores it in the register $R_A = (K_A,T_A)$ where $K_A$ stores the first $\elloutput$ bits and $T_A$ stores the last $\ellEV$ bits. Similarly, Bob computes $f_h(\bS_B)$ and stores it in the register $R_B = (K_B, T_B)$.

\item \textbf{Checking bits.} 
Alice announces $T_A$, and Bob announces the flag $ \FEV:=\one[T_A=T_B]$. If $\FEV=1$ the
final keys are $K_A$ and $K_B$ of length $\elloutput$; if $\FEV=0$, Alice and Bob replace $K_A,K_B$ with the special abort symbol  
$\perp$. 
\end{enumerate}
\end{prot}

\subsection{Definitions}\label{subsec:definitions}

 We begin by stating the security definition of a QKD protocol. 
  Let $\rho_{K_AK_B\Efin}$ denote the final state of such a protocol,
where $K_A$ and $K_B$ are Alice's and Bob's final key registers and
$\Efin$ contains all information available to Eve at the end of the
protocol, including all public communication. The key registers
may contain keys of different lengths; equivalently, they have the
direct-sum structure
\begin{equation}
  K_A=\bigoplus_{m\in\mathbb{N}_0}K_A^m,
  \qquad
  K_B=\bigoplus_{m\in\mathbb{N}_0}K_B^m,
\end{equation}
where $K_A^m$ and $K_B^m$ store $m$-bit keys. The case $m=0$
corresponds to the abort symbol $\perp$. 

\begin{definition}[Security of QKD]\label{def:qkdsecurity}
Consider a QKD protocol, with output state $\rho_{K_A K_B \Efin}$. For every $m\in\mathbb{N}_0$, let $\Omfin{m}$ denote the event that
Alice and Bob both output keys of length $m$. Then, the protocol  is
$\varepsilon^{\mathrm{secure}}$-secure if, for all possible output states $\rho_{K_A K_B \Efin}$,
\begin{equation}
  \sum_{m\geq 1}
  \dist{
    \rho_{K_AK_B\Efin\wedge\Omfin{m}}
  }{
    \tau_{K_A K_B}^{(m)}
    \otimes
    \rho_{\Efin\wedge\Omfin{m}}
  }
  \leq
  \varepsilon^{\mathrm{secure}}.
\end{equation}
where  $\tau^{(m)}_{K_A K_B} \coloneq 2^{-m}\sum_{r\in\{0,1\}^m} \ketbra{ r r}_{K_A K_B}$ denotes the ideal key state on the output key registers.
\end{definition}

We next recall the relevant definitions of hash families.

\begin{definition}[$\delta$-almost-universal$_2$,
universal$_2$, and ideal universal$_2$ hash families]
\label{def:2universal}
Let $D$ be a domain and let $D'$ be a finite codomain. A hash
family consists of a set $\mathcal{H}$ of functions $h:D\to D'$,
together with a specified probability distribution over
$\mathcal{H}$. The family is
$\delta$-almost-universal$_2$ ($\delta$-AU$_2$) if, for every pair
of distinct elements $x,y\in D$,
\begin{equation}
  \Pr_h\!\left[h(x)=h(y)\right]
  \leq \delta,
  \label{eq:AU2-definition}
\end{equation}
where the probability is taken over the choice of $h$. The family is universal$_2$ if
\begin{equation}
  \Pr_h\!\left[h(x)=h(y)\right]
  \leq \frac{1}{|D'|}
\end{equation}
for every pair of distinct elements $x,y\in D$. 
\end{definition}

For the combined protocol, it is not sufficient to only consider 
collisions of the complete hash outputs. We must also control
collisions of the $\ellEV$ bits that are compared during EV. This motivates the following definition.

\begin{definition}[Truncated $(\ellEV,\delta)$-AU$_2$ property]
\label{def:tail-au2}
The hash families used in the protocol from \cref{subsec:protocol} satisfy the
\emph{truncated $(\ellEV,\delta)$-AU$_2$ property} if, for every admissible
choice of $\ellprePA$ and $\elloutput$, and every pair of distinct
strings $
  s,s'\in\{0,1\}^{\ellprePA}$, we have
\begin{equation}
  \Pr_h\!\left[
    \operatorname{last}_{\ellEV}\!\left(f_h(s)\right)
    =
    \operatorname{last}_{\ellEV}\!\left(f_h(s')\right)
  \right]
  \leq \delta.
  \label{eq:tail-au2}
\end{equation}
Here, $\operatorname{last}_{\ellEV}$ selects the last $\ellEV$
bits of the $(\elloutput+\ellEV)$-bit hash output.
The probability is taken over the choice of the hash seed $h$
according to the distribution specified by the protocol.
\end{definition}

\section{Results}\label{sec:results}

For the QKD protocol with the combined PA and EV step, we define the event $
  \Ommatch
  := \{\bS_A=\bS_B\}$ when the pre-amplification strings match, and $
  \OmEVpost
  := \{T_A=T_B\}
   = \{\FEV=1\}$ when the error-verification check passes. 
  Note that for our protocol, $
  \Efin:=\EprePA\HPA T_A\FEV$
denote Eve's information at the end of the combined protocol.
Recall that $\EprePA$ already includes all public communication
preceding the combined protocol, including the register $F_K$. We can now state and prove the following theorem. 

\begin{theorem}[Security of QKD protocol with combined EV and PA step]
\label{thm:main}

Let $\rho_{R_A\EprePA\HPA}$ denote the state immediately after
the hash function has been applied and before $T_A$ and $\FEV$
are announced.
Suppose that for all such possible states,  we have  
\begin{equation}
\begin{aligned}
  &\sum_{\elloutput\geq 1}
  d \Big(
    \rho_{R_A\EprePA\HPA
      \wedge \OmFK{\elloutput} \wedge\Ommatch}, \\     
 &   \tau^{(\elloutput+\ellEV)}_{R_A}
    \otimes
    \rho_{\EprePA\HPA
      \wedge \OmFK{\elloutput}  \wedge\Ommatch} \Big) \\ 
 & \leq\varepsilon,
  \label{eq:matching-secrecy}
\end{aligned}
\end{equation}
where $\OmFK{\elloutput}$ denotes the event that $F_K$ stores $\elloutput$. That is, a modified notion of secrecy holds. 

Suppose that the hash family used satisfies the truncated $(\ellEV,\delta)$-AU$_2$ property. Then, 
the overall protocol is 
$\bigl(2\varepsilon+\delta)$-secure, that is:
\begin{equation}
  \sum_{m\geq 1}
  \dist{
    \rho_{K_AK_B\Efin\wedge\Omfin{m}}
  }{
    \tau_{K_AK_B}^{(m)}
    \otimes
    \rho_{\Efin\wedge\Omfin{m}}
  }
  \leq
  2\varepsilon+\delta.
  \label{eq:combined-secrecy}
\end{equation}
is satisfied for all output states $ \rho_{K_AK_B\Efin}$.
\end{theorem}

\begin{proof}
The proof follows from the following chain of inequalities (explained below). We start with the quantity that we are required to bound:

\begin{widetext}

\begin{align}
&\sum_{\elloutput\geq1}
  \dist{
    \rho_{K_AK_B\Efin\wedge\Omfin{\elloutput}}
  }{
    \tau_{K_AK_B}^{(\elloutput)}
    \otimes
    \rho_{\Efin\wedge\Omfin{\elloutput}}
  }
  \label{eq:proof-start}
\\
&\leq
  \sum_{\elloutput\geq1}
  \dist{
    \rho_{K_AK_B\Efin
      \wedge\Omfin{\elloutput}\wedge\Ommatch}
  }{
    \tau_{K_AK_B}^{(\elloutput)}
    \otimes
    \rho_{\Efin
      \wedge\Omfin{\elloutput}\wedge\Ommatch}
  }
  \nonumber\\[-0.25em]
&\qquad+
  \sum_{\elloutput\geq1}
  \dist{
    \rho_{K_AK_B\Efin
      \wedge\Omfin{\elloutput}\wedge\Ommatch^{\mathrm C}}
  }{
    \tau_{K_AK_B}^{(\elloutput)}
    \otimes
    \rho_{\Efin
      \wedge\Omfin{\elloutput}\wedge\Ommatch^{\mathrm C}}
  }
  \label{eq:proof-split}
\\
&\leq
  \sum_{\elloutput\geq1}
  \dist{
    \rho_{K_AK_B\Efin
      \wedge\Omfin{\elloutput}\wedge\Ommatch}
  }{
    \tau_{K_AK_B}^{(\elloutput)}
    \otimes
    \rho_{\Efin
      \wedge\Omfin{\elloutput}\wedge\Ommatch}
  }
  +\delta
  \label{eq:proof-bound-collisions}
\\
&=
  \sum_{\elloutput\geq1}
  \dist{
    \rho_{K_AT_A\EprePA\HPA
      \wedge\OmFK{\elloutput}\wedge\Ommatch}
  }{
    \tau_{K_A}^{(\elloutput)}
    \otimes
    \rho_{T_A\EprePA\HPA
      \wedge\OmFK{\elloutput}\wedge\Ommatch}
  }
  +\delta
  \label{eq:proof-remove-copy}
\\
&\leq
  \sum_{\elloutput\geq1}
  \dist{
    \rho_{K_AT_A\EprePA\HPA
      \wedge\OmFK{\elloutput}\wedge\Ommatch}
  }{
    \tau_{K_A}^{(\elloutput)}
    \otimes
    \tau_{T_A}^{(\ellEV)}
    \otimes
    \rho_{\EprePA\HPA
      \wedge\OmFK{\elloutput}\wedge\Ommatch}
  }
  \nonumber\\[-0.25em]
&\qquad+
  \sum_{\elloutput\geq1}
  \dist{\tau_{K_A}^{(\elloutput)}  \otimes 
    \tau_{T_A}^{(\ellEV)}
    \otimes
    \rho_{\EprePA\HPA
      \wedge\OmFK{\elloutput}\wedge\Ommatch}
  }{\tau_{K_A}^{(\elloutput)} \otimes 
    \rho_{T_A\EprePA\HPA
      \wedge\OmFK{\elloutput}\wedge\Ommatch}
  }
  +\delta
  \label{eq:proof-triangle}
\\
&\leq
  2\sum_{\elloutput\geq1}
  \dist{
    \rho_{R_A\EprePA\HPA
      \wedge\OmFK{\elloutput}\wedge\Ommatch}
  }{
    \tau_{R_A}^{(\elloutput+\ellEV)}
    \otimes
    \rho_{\EprePA\HPA
      \wedge\OmFK{\elloutput}\wedge\Ommatch}
  }
  +\delta
  \label{eq:proof-discard-key}
\\
&\leq
  2\varepsilon+\delta.
  \label{eq:proof-final}
\end{align}
\end{widetext}

To obtain \cref{eq:proof-split}, we split each partial state
according to whether $\Ommatch$ occurs (or not) and apply the triangle
inequality. 
For \cref{eq:proof-bound-collisions}, fix distinct values of
$\bS_A$ and $\bS_B$, together with all other values fixed before
the hash seed is chosen. Since the seed is sampled independently,
the truncated $\delta$-AU$_2$ property implies that $T_A=T_B$
with probability at most $\delta$. Thus, the probability of the protocol outputting a non-zero key while the pre-amplification strings differ is upper bounded by $\delta$. 
To obtain \cref{eq:proof-remove-copy}, we use the fact that for the combined protocol, $\Ommatch \wedge \Omfin{\elloutput} = \Ommatch  \wedge \OmFK{\elloutput}$, and the fact that conditioned on $\Ommatch$, $K_AT_A$ and $K_BT_B$ store identical strings. Thus, we do not need to add the register $\FEV$ since it stores a deterministic value conditioned on $\Ommatch$.

Equation~\eqref{eq:proof-triangle} follows by inserting a state
in which both $K_A$ and $T_A$ are uniform and applying the
triangle inequality. The second trace distance in that expression is 
obtained from the first distance by tracing out $K_A$. Thus by using the fact that CPTP maps cannot increase trace distance, we obtain \cref{eq:proof-discard-key}.
Finally, \cref{eq:proof-final} follows directly from the first
hypothesis of \cref{thm:main}.  This concludes our proof.
\end{proof}

\subsection{Impact on key rates}
\label{subsec:impactonkeyrates}

The secrecy statement required in \cref{eq:matching-secrecy} is
closely related to the secrecy statement normally proved for a protocol
without the final EV step. The difference is that the state in
\cref{eq:matching-secrecy} is partial on the event $\Ommatch$. One may wonder whether this conditioning on $\Ommatch$ can be removed; we show in Appendix~\ref{app:random-matrix-counterexample} a counter example where the state before hashing (without conditioning on $\Ommatch$) is perfectly secret, but Eve learns one bit of information about the final key with non-negligible probability for suitable parameter choices.  Thus, we do not expect to get rid of the (subnormalized) conditioning on $\Ommatch$ in \cref{thm:main} without imposing additional assumptions. 

Thankfully, for proofs based on the leftover hashing lemma for smooth min entropy \cite{TomamichelLeverrier2017} or {R\'enyi} entropy \cite{dupuisentropy}, this does not pose a problem. For simplicity, we only briefly describe how those proofs can be suitably adapted to \cref{thm:main}, since a full analysis requires a number of  details that would distract us from the matter at hand. Note that such proofs require the hashing to be universal$_2$, on top of the truncated $(\ellEV,\delta)$-AU$_2$ property required by \cref{thm:main}. Toeplitz hashing, a widely used procedure, satisfies both properties (see Appendix~\ref{app:toeplitz-truncated-au2}). In our description below, we use $\sigma$ to refer to the states corresponding to a protocol implementing separate EV and PA, and $\rho$ to refer to states corresponding to a protocol which implements the combined EV and PA step. Note that the events $\Omega_\mathrm{EV}$, and registers $\FEV$ are slightly different in the two protocols.

Such proofs consider different events $\Omega(\cobs)$ corresponding to different observations $\cobs$ during the protocol. These observations determine the length $\ell(\cobs) = \elloutput$ of the output key that can be produced. Then, the leftover hashing lemma is applied on each such state separately, resulting in an expression that looks like:
\begin{equation} \label{eq:temp1}
  \begin{aligned}
&\sum_{\cobs} \Pr(\Omega(\cobs) \wedge \Omega_\mathrm{EV}) \times \\
& 2^{t_1 \left( \mathbb{H}(\bS_A| \EprePA \HEV \CEV \FEV)_{\sigma | \Omega(\cobs) \wedge \Omega_\mathrm{EV}}-\ell(\cobs) \right) +t_2}
  \end{aligned}
\end{equation}
where $t_1$ and $t_2$ are constants and $\mathbb{H}$ is the appropriate entropy\footnote{If one uses smooth-min entropy, one also obtains an additive constant which we omit here, since it does not affect our analysis}, and $\sigma_{\bS_A \EprePA}$ is the state just before privacy amplification and error-verification, $\HEV$ and $\CEV$ denote classical registers storing the hash choice and hash output of the separately implemented EV procedure, and $\FEV$ denotes the communicated bit indicating success or failure of EV. Thus, Alice and Bob additionally communicate $\HEV, \CEV, \FEV$ before performing the separate PA step. Here $\Omega_\mathrm{EV}$ is the event that error-verification (when implemented separately, before PA) has passed.

This expression is first simplified by noting that the register $\FEV$ takes on a deterministic value when conditioned on $\Omega_\mathrm{EV}$ and can be safely ignored. Then, the conditioning on $\Omega_\mathrm{EV}$ is removed via suitable properties of the entropy (see for e.g. \cite{kamin2025renyisecurityframeworkcoherent,tupkary2026rigoroussecurityproofspractical,TupkaryEtAl2026,TomamichelLeverrier2017}). After doing so, we obtain
\begin{equation} \label{eq:temp2}
  \begin{aligned}
&\sum_{\cobs} \Pr(\Omega(\cobs) ) \times \\
& 2^{t_1 \left( \mathbb{H}(\bS_A| \EprePA \HEV \CEV )_{\sigma | \Omega(\cobs)}-\ell(\cobs) \right) +t_2}
  \end{aligned}
\end{equation}
This is further simplified by removing the $\HEV$ register without penalty (since it is independent of the rest of the registers using data processing \cite{tomamichel_quantum_2016}), and removing $\CEV$ via suitable chain rules, yielding 
\begin{equation} \label{eq:temp3}
  \begin{aligned}
&\sum_{\cobs} \Pr(\Omega(\cobs) ) \times \\
& 2^{t_1 \left( \mathbb{H}(\bS_A| \EprePA )_{\sigma | \Omega(\cobs) }-\ell(\cobs) -\ellEV \right) +t_2}
  \end{aligned}
\end{equation}

This quantity is then suitably bounded via the QKD security proof. 

When we wish to analyze the combined protocol, we will obtain, analogous to \cref{eq:temp3},
\begin{equation} \label{eq:combined-entropy-bound}
  \begin{aligned}
& \sum_{\cobs} \Pr(\Omega(\cobs) \wedge \Ommatch )  \times \\
&2^{t_1 \left( \mathbb{H}(\bS_A| \EprePA \HPA)_{\rho | \Omega(\cobs) \wedge \Ommatch}- \ell(\cobs) -\ellEV \right) +t_2}.
  \end{aligned}
\end{equation}
Notice that the output length of the combined hash is $(\elloutput+\ellEV)$. Since the register $\HPA$ storing the hash seed is completely independent of the rest of the registers (even conditioned on the events), it can be removed without penalty using the data processing inequality \cite{tomamichel_quantum_2016}.  Then, we can use the same procedure as before to get rid of the event $\Ommatch$, and instead analyze 
\begin{equation} \label{eq:combined-entropy-bound-with-seed}
\sum_{\cobs} \Pr(\Omega(\cobs) )  2^{t_1 \left( \mathbb{H}(\bS_A| \EprePA )_{\rho | \Omega(\cobs)} - \ell(\cobs) -\ellEV \right) +t_2}.
\end{equation}
This can now be bounded in the same manner as \cref{eq:temp3}, using standard QKD security proof methods.

The above justification can be applied to proofs \cite{tupkary2026rigoroussecurityproofspractical,TupkaryEtAl2026} based on the postselection technique (via leftover hashing lemmas) \cite{nahar_postselection_2024,christandl_postselection_2009}, entropic uncertainty relations \cite{tomamichel_uncertainty_2011}, and entropy accumulation theorems \cite{dupuis_entropy_2020,metger_generalised_2022,arqand_marginal_2025}.

\begin{corollary}[Informal]
Consider the key rates computed in Ref.~ \cite{tupkary2026rigorouscompletesecurityproof}, which presents a rigorous treatment of QKD security. Fix $\varepsilon_\mathrm{PA},\varepsilon_\mathrm{EV} \in [0,1]$, and consider the key rates $\ell(\cobs)$ obtained in that work, for a protocol with separate EV and PA steps. The same key rate is valid when the EV and PA steps are combined, using Toeplitz hashing, as described in \cref{subsec:protocol}, for an overall security parameter of $2\varepsilon_\mathrm{PA}+\varepsilon_\mathrm{EV}$.
\end{corollary}

However, it seems unclear whether proofs based on ``virtual phase error correction'' \cite{koashi_simple_2009} have this property. In particular, the fundamental logic behind such proofs is to use the privacy amplification procedure on Alice's side to construct a ``virtual phase error correction protocol'', in which the success probability of the virtual protocol can be used to derive properties about the real state between Alice's final key and Eve's side-information. However, these properties are typically isolated solely to the state of Alice's final key and Eve's side-information, without any conditioning on events determined by Bob's error-corrected value. Hence it is unclear if such proofs are compatible with the approach described in this work. 

We observe that in Ref.~\cite{TsurumaruEtAl2025}, to handle error-verification loopholes in virtual phase error correction proofs,
it was proposed to use phase-error estimates to bound the smooth min-entropy instead of the success probability of virtual phase error correction, and invoke event-conditioning properties of the former to obtain the desired security properties.  From our perspective, this suggests that the resulting proof is more naturally interpreted as being based on smooth min-entropy. Accordingly, the conclusions obtained in this approach appear to rely on entropic arguments of the type presented above, rather than following directly from the virtual phase error correction picture itself.

\begin{remark} \label{remark:AlicevsBob}
Note that either Alice or Bob can do the first two steps of the protocol from \cref{subsec:protocol}; The third step is the same for both parties.
In the fourth step, one party announces their last $\ellEV$ bits, and the other party only announces whether the bits matched or not.  This is the only step where the actions of Alice and Bob are not symmetric. We choose Alice to announce her bits, and in \cref{thm:main}, reduce the analysis to secrecy wrt Alice's key. We could, without loss of generality, have chosen Bob in both places. 
\end{remark}

\section{Conclusion}\label{sec:conclusion}
In this work, we showed that the security of a QKD protocol with a combined
EV--PA step can be obtained directly from a slightly modified secrecy statement
for the protocol before the final EV announcement. Beyond this secrecy
statement, our result only requires the hashing procedure to satisfy the
property in \cref{def:tail-au2}. We showed that random Toeplitz hashing, which is
widely used in QKD, satisfies this property. We also explained how standard
entropy-based security proofs using the leftover hashing lemma can be adapted
to establish the required secrecy statement.

Our result therefore shows how a QKD protocol with a combined EV--PA step can be treated
rigorously in practical security proofs, without any modification to the key rate expressions,  and with only a minor worsening of the security
parameter.

\begin{acknowledgments}
We used OpenAI's GPT models and Anthropic's Claude models during
the development and preparation of this work. The initial proof
approach and the anticipated factor of two in the secrecy bound
were supplied by the authors. The requirement for additional conditioning on $\Ommatch$ was identified by the models, along with the counter example in Appendix~\ref{app:random-matrix-counterexample}. They also assisted with
drafting and editing the manuscript, which was primarily written by DT. All authors are responsible for the entirety of this paper.

This work was funded by the NSERC Alliance QUINT and
DND-MicroNet. It was primarily conducted at the Institute for Quantum Computing, University of Waterloo, which is funded
by the Government of Canada through ISED. DT acknowledges the support of the Natural Sciences and Engineering
Research Council of Canada (NSERC) (ALLRP-578455-2022). ZW is funded by the Mike and Ophelia Lazaridis Fellowship. DT thanks Hagai Eisenberg, Yehuda Pilnak for drawing his attention to this problem. 
\end{acknowledgments}

\bibliography{references_verified}

\appendix

\section{Counterexample using random full-row-rank matrices}
\label{app:random-matrix-counterexample}
One may attempt to prove a variation of \cref{thm:main} which does not have the conditioning on $\Ommatch$ in the required secrecy statement. We give a counterexample based on a standard family of random linear
hash functions. It shows that secrecy of the complete hash output
before error verification does not imply security of the combined
protocol unless the secrecy statement involves states partial on
$\Ommatch$.

Let
\begin{equation}
 r\coloneq \elloutput+\ellEV,
 \qquad
 n>r,
\end{equation}
and let $\mathcal{A}_{r,n}$ be the set of all full-row-rank matrices in
$\mathbb{F}_2^{r\times n}$. The hash seed is a matrix
$A\in\mathcal{A}_{r,n}$ sampled uniformly, and
\begin{equation}
 f_A(s)\coloneq As.
 \label{eq:full-rank-hash-family}
\end{equation}
The first $\elloutput$ output bits are retained, while the last
$\ellEV$ bits are announced.

\begin{lemma}[Random full-row-rank counterexample]
\label{lem:full-rank-counterexample}
The family in \cref{eq:full-rank-hash-family} has the following
properties:
\begin{enumerate}
\item it is universal$_2$;
\item its last $\ellEV$ output bits are
      $(\ellEV,2^{-\ellEV})$-almost-universal$_2$;
\item there is a state for which Alice's complete hash output
      $R_A=(K_A,T_A)$ is perfectly uniform and independent of the
      public seed before error verification; but
\item the output of the combined protocol satisfies
\begin{equation}
\begin{aligned}
&\dist{
 \rho_{K_AK_B\Efin\wedge\Omfin{\elloutput}}
}{
 \tau_{K_AK_B}^{(\elloutput)}
 \otimes
 \rho_{\Efin\wedge\Omfin{\elloutput}}
}
\\
&\hspace{3em}\geq
 \frac{2^r-2^{\ellEV}}{4(2^n-1)}.
\end{aligned}
\label{eq:full-rank-final-lower-bound}
\end{equation}
\end{enumerate}
\end{lemma}

\begin{proof}

\textbf{Hashing properties: }
We first verify the two hashing properties. Fix a nonzero
$z\in\mathbb{F}_2^n$.  For showing it is universal$_2$, we need to bound $\Pr_A[Az=0]$.

Notice that the subspace
\[
z^\perp
\coloneq
\{u\in\mathbb{F}_2^n:u^{\mathsf T}z=0\}
\]
contains $2^{n-1}$ vectors. Suppose that $i$ linearly independent rows of $A$ have already been chosen in
$z^\perp$. Their span contains $2^i$ vectors. The next row has
$2^n-2^i$ possible values outside this span, of which
$2^{n-1}-2^i$ lie in $z^\perp$. Therefore,
\begin{align}
\Pr_A[Az=0]
&=
\prod_{i=0}^{r-1}
\frac{2^{n-1}-2^i}{2^n-2^i}
\nonumber\\
&=
\frac{2^{n-r}-1}{2^n-1}
<2^{-r}.
\label{eq:full-rank-complete-collision}
\end{align}
This proves that the complete hash family is universal$_2$. The same proof applies to the last $\ellEV$ rows of $A$ (where we can imagine that we sample the rows in reverse order). 

\textbf{State:} We will now specify the state that is hashed.
Let $S_{A}$ be uniform on $\mathbb{F}^n_2$, and let $S_{A,1}$ be the first bit and $
S_{A,2}^n
\coloneq
(S_{A,2},\ldots,S_{A,n})$ be the remaining bits. Define 
\begin{equation}
\bS_A=(S_{A,1},S_{A,2}^n),
\qquad
\bS_B=(0,S_{A,2}^n).
\label{eq:full-rank-counterexample-state}
\end{equation}
Thus $\Ommatch$ is the event $S_{A,1}=0$. 
Since $\bS_A$ is uniform, it is easy to show that
$A\bS_A$ is uniform for every choice of $A$. Hence
\begin{equation}
\rho_{R_A\HPA}
=
\tau_{R_A}^{(r)}
\otimes
\rho_{\HPA}.
\label{eq:full-rank-plain-secrecy}
\end{equation}
In particular, the ordinary pre-verification secrecy error is zero.

\textbf{Eve learns one bit:}  The idea is to identify public
seeds for which acceptance reveals one nontrivial parity of Alice's
retained key.

Write
\begin{equation}
A=
\begin{pmatrix}
a&\widetilde{A}
\end{pmatrix},
\qquad
a\in\mathbb{F}_2^r,
\qquad
\widetilde{A}\in\mathbb{F}_2^{r\times(n-1)},
\end{equation}
and split $a=(a_K,a_T)$ according to the retained and announced
output bits. Alice's and Bob's outputs are
\begin{equation}
R_A
=
S_{A,1}a+\widetilde{A}S_{A,2}^n,
\qquad
R_B
=
\widetilde{A}S_{A,2}^n.
\label{eq:full-rank-counterexample-outputs}
\end{equation}
They therefore accept if and only if
\begin{equation}
S_{A,1}a_T=0.
\label{eq:full-rank-acceptance}
\end{equation}

Define
\begin{equation}
\Gamma
\coloneq
\left\{
A\in\mathcal{A}_{r,n}:
\begin{array}{c}
\text{there exists }q=(q_K,q_T)\in\mathbb{F}_2^r
\text{ such that}\\[1mm]
q_K\neq0
\text{ and }
q^{\mathsf T}A=e_1^{\mathsf T}
\end{array}
\right\},
\label{eq:full-rank-bad-seeds}
\end{equation}
where $e^{\mathsf T}_1=(1,0,\ldots,0)$. 
Hence, for every $A\in\Gamma$,
\begin{equation}
q_K^{\mathsf T}K_A+q_T^{\mathsf T}T_A
=
q^{\mathsf T}R_A
=
q^{\mathsf T}A\bS_A
=
S_{A,1}.
\label{eq:full-rank-output-parity}
\end{equation}

We first describe what happens for such a seed. 

Suppose first that $a_T\neq0$. By
\cref{eq:full-rank-acceptance}, acceptance is equivalent to
$S_{A,1}=0$. Equation~\eqref{eq:full-rank-output-parity} then gives
\begin{equation}
q_K^{\mathsf T}K_A
=
q_T^{\mathsf T}T_A.
\label{eq:full-rank-revealed-parity}
\end{equation}
The seed determines $q$ (via $A$), and $T_A$ is publicly announced. Thus Eve
learns the nontrivial parity $q_K^{\mathsf T}K_A$ of the retained key. Therefore, the key is no longer uniform and it is easy to show that the distance of the final state from uniform is 
$ 1/2$. Since acceptance
occurs with probability $1/2$, each such seed contributes $1/4$ to
the subnormalized security distance.

If instead $a_T=0$, the protocol always accepts. Recall that $q_K^{\mathsf T}a_K=1$, so $a_K \neq 0$. Whenever $S_{A,1}=1$, Alice's and Bob's retained keys
differ by $a_K$ (\cref{eq:full-rank-counterexample-outputs}). This occurs with probability $1/2$, so each such
seed contributes at least $1/2$ to the trace distance.

Averaging the two contributions gives
\begin{align}
&d\left(
\rho_{K_AK_BE_{\mathrm{fin}}}\wedge\Omega_{\elloutput}^{\mathrm{fin}},
\tau_{K_AK_B}^{(\elloutput)}
 \otimes
\rho_{E_{\mathrm{fin}}}\wedge\Omega_{\elloutput}^{\mathrm{fin}}
\right) \nonumber \\
& \geq  
\frac{1}{4}\Pr_A[A\in\Gamma,\ a_T\neq0]
+
\frac{1}{2}\Pr_A[A\in\Gamma,\ a_T=0]
\nonumber\\
& \geq  
\frac{1}{4}\Pr_A[A\in\Gamma]
\label{eq:full-rank-final-lower-bound}
\end{align}

\textbf{Probability of bad seeds:}
For a fixed full-row-rank matrix $A$, define
\[
\mathcal{V}(A)
\coloneq
\left\{
q^{\mathsf T}A:
q=(q_K,q_T)\in\mathbb{F}_2^r,\ q_K\neq0
\right\}.
\]
The map $q\mapsto q^{\mathsf T}A$ is injective. Indeed, if
$q^{\mathsf T}A=(q')^{\mathsf T}A$, then $
(q-q')^{\mathsf T}A=0$,
and the full row rank of $A$ implies $q=q'$.

There are $2^{\elloutput}-1$ nonzero choices for $q_K$ and
$2^{\ellEV}$ arbitrary choices for $q_T$. Therefore,
\begin{equation}
|\mathcal{V}(A)|
=
(2^{\elloutput}-1)2^{\ellEV}
=
2^r-2^{\ellEV}.
\label{eq:full-rank-number-of-bad-vectors}
\end{equation}

The event $\Gamma$ is precisely the event
\[
e_1^{\mathsf T}\in\mathcal{V}(A).
\]
For a uniformly random full-row-rank matrix $A$, every nonzero vector
in $\mathbb{F}_2^n$ has the same probability of belonging to
$\mathcal{V}(A)$. Since $\mathcal{V}(A)$ always contains
$2^r-2^{\ellEV}$ of the $2^n-1$ nonzero vectors, it follows that
\begin{equation}
\Pr_A[\Gamma]
=
\frac{2^r-2^{\ellEV}}{2^n-1}.
\label{eq:full-rank-bad-seed-probability}
\end{equation}
This concludes our proof.
\end{proof}

\section{Truncated universality of Toeplitz hashing}
\label[appendix]{app:toeplitz-truncated-au2}

\begin{lemma}
\label{lem:toeplitz-truncated-au2}
Suppose that $f_h(s)=M_hs$, where $M_h$ is a uniformly random binary
Toeplitz matrix of size
$(\elloutput+\ellEV)\times\ellprePA$. Then the family satisfies the
truncated $(\ellEV,2^{-\ellEV})$-AU$_2$ property.
\end{lemma}

\begin{proof}
Set
\begin{equation}
m\coloneq\elloutput+\ellEV,
\qquad
n\coloneq\ellprePA.
\end{equation}
An $m\times n$ binary Toeplitz matrix is specified by its $m+n-1$
diagonal entries and hence there are $2^{m+n-1}$ such matrices.

Fix a nonzero $x\in\mathbb{F}_2^n$ and an arbitrary
$y\in\mathbb{F}_2^m$. We claim that exactly $2^{n-1}$ Toeplitz
matrices satisfy
\begin{equation}
Mx=y.
\label{eq:toeplitz-linear-system}
\end{equation}

Write $M_{ij}=u_{i-j}$, where $u_{1-n},\ldots,u_{m-1}\in\mathbb F_2$ are the bits specifying the diagonals of $M$.
Let $j_0$ be the first nonzero coordinate of $x$. After fixing
the $n-1$ diagonal entries other than
\begin{equation}
u_{1-j_0},u_{2-j_0},\ldots,u_{m-j_0},
\end{equation}
the $m$ equations in \cref{eq:toeplitz-linear-system} determine these
remaining $m$ entries uniquely, one equation at a time. Thus there is
one solution for every assignment of the other $n-1$ diagonal
entries, giving exactly $2^{n-1}$ solutions.

Consequently,
\begin{equation}
\Pr_M[Mx=y]
=
\frac{2^{n-1}}{2^{m+n-1}}
=
2^{-m}.
\end{equation}
Therefore, $Mx$ is uniform on $\mathbb{F}_2^m$ for every nonzero $x$.

Now fix $s\neq s'$ and set $x=s\oplus s'\neq0$. Since $Mx$ is
uniform, its last $\ellEV$ bits are uniform on
$\mathbb{F}_2^{\ellEV}$. Hence
\begin{align}
&\Pr_h\!\left[
\operatorname{last}_{\ellEV}\!\left(f_h(s)\right)
=
\operatorname{last}_{\ellEV}\!\left(f_h(s')\right)
\right]
\nonumber\\
&\qquad=
\Pr_h\!\left[
\operatorname{last}_{\ellEV}(M_hx)=0
\right]
=
2^{-\ellEV},
\end{align}
as required.
\end{proof}

\end{document}